\documentclass{sig-alternate}

\begin{document}


\newdef{theorem}{Theorem}
\newdef{lemma}{Lemma}
\newdef{prop}{Proposition}
\newdef{cor}{Corollary}
\newdef{definition}{Definition}
\newdef{claim}{Claim}
\newdef{problem}{Problem}
\newdef{remark}{Remark}
\newdef{fact}{Fact}
\newdef{example}{Example}

\newenvironment{proofof}[1]{\noindent{\bf Proof of #1:}}{\qed\\}


\newcommand{\comment}[1]{\begin{quote}\sf [*** #1 ***]\end{quote}}
\newcommand{\tinyspace}{\mspace{1mu}}
\newcommand{\microspace}{\mspace{0.5mu}}
\newcommand{\op}[1]{\operatorname{#1}}
\newcommand{\snorm}[1]{\lVert\tinyspace#1\tinyspace\rVert}
\newcommand{\ceil}[1]{\left\lceil #1 \right\rceil}
\newcommand{\floor}[1]{\left\lfloor #1 \right\rfloor}
\def\iso{\cong}
\newcommand{\defeq}{\stackrel{\mathrm{def}}{=}}
\newcommand{\tr}{\operatorname{Tr}}
\newcommand{\rank}{\operatorname{rank}}
\renewcommand{\det}{\operatorname{Det}}
\renewcommand{\vec}{\operatorname{vec}}
\newcommand{\im}{\operatorname{Im}}
\renewcommand{\t}{{\scriptscriptstyle\mathsf{T}}}
\newcommand{\ip}[2]{\left\langle #1 , #2\right\rangle}
\newcommand{\sip}[2]{\langle #1 , #2\rangle}
\def\({\left(}
\def\){\right)}
\def\I{\mathbb{1}}

\newcommand{\triplenorm}[1]{
  \left|\!\microspace\left|\!\microspace\left| #1
  \right|\!\microspace\right|\!\microspace\right|}

\newcommand{\fid}{\operatorname{F}}
\newcommand{\setft}[1]{R^{#1}}
\newcommand{\lin}[1]{\setft{L}\left(#1\right)}
\newcommand{\density}[1]{\setft{D}\left(#1\right)}
\newcommand{\unitary}[1]{\setft{U}\left(#1\right)}
\newcommand{\trans}[1]{\setft{T}\left(#1\right)}
\newcommand{\herm}[1]{\setft{Herm}\left(#1\right)}
\newcommand{\pos}[1]{\setft{Pos}\left(#1\right)}
\newcommand{\pd}[1]{\setft{Pd}\left(#1\right)}
\newcommand{\sphere}[1]{\mathcal{S}\!\left(#1\right)}
\newcommand{\opset}[3]{\setft{#1}_{#2}\!\left(#3\right)}

\def\complex{\mathbb{C}}
\def\real{\mathbb{R}}
\def\natural{\mathbb{N}}
\def\integer{\mathbb{Z}}

\def \lket {\left|}
\def \rket {\right\rangle}
\def \lbra {\left\langle}
\def \rbra {\right|}
\newcommand{\ket}[1]{\lket\microspace #1 \microspace\rket}
\newcommand{\bra}[1]{\lbra\microspace #1 \microspace\rbra}

\newenvironment{mylist}[1]{\begin{list}{}{
    \setlength{\leftmargin}{#1}
    \setlength{\rightmargin}{0mm}
    \setlength{\labelsep}{2mm}
    \setlength{\labelwidth}{8mm}
    \setlength{\itemsep}{0mm}}}
    {\end{list}}

\newcommand{\class}[1]{\textup{#1}}
\newcommand{\reg}[1]{\mathsf{#1}}

\newenvironment{namedtheorem}[1]
    {\begin{trivlist}\item {\bf #1.}\em}{\end{trivlist}}

\def\X{\mathcal{X}}
\def\Y{\mathcal{Y}}
\def\Z{\mathcal{Z}}
\def\W{\mathcal{W}}
\def\A{\mathcal{A}}
\def\B{\mathcal{B}}
\def\V{\mathcal{V}}
\def\U{\mathcal{U}}
\def\C{\mathcal{C}}
\def\D{\mathcal{D}}
\def\E{\mathcal{E}}
\def\F{\mathcal{F}}
\def\M{\mathcal{M}}
\def\R{\mathcal{R}}
\def\P{\mathcal{P}}
\def\G{\mathcal{G}}
\def\Q{\mathcal{Q}}
\def\S{\mathcal{S}}
\def\T{\mathcal{T}}
\def\K{\mathcal{K}}
\def\O{\mathcal{O}}
\def\J{\mathcal{J}}
\def\yes{\text{yes}}
\def\no{\text{no}}
\def\ve{\varepsilon}
\def\poly{\mathrm{poly}}
\def\polylog{\mathrm{polylog}}
\def\opt{\mathrm{opt}}
\def\thr{\mathsf{thr}}
\def\Diag{\mathrm{Diag}}


\newcommand{\rec} [3] {\mathrm{rec}^{#1}_{#2}\left( #3 \right)}
\newcommand{\lrec} [3] {\widetilde{\mathrm{rec}}^{#1}_{#2}\left( #3 \right)}
\newcommand{\srec} [3] {\mathrm{srec}^{#1}_{#2}\left( #3 \right)}
\newcommand{\lsrec} [3] {\widetilde{\mathrm{srec}}^{#1}_{#2}\left( #3 \right)}
\newcommand{\fsrec} [3] {\widehat{\mathrm{srec}}^{#1}_{#2}\left( #3 \right)}
\newcommand{\err} [2] {\mathrm{err}_{#1}\left( #2 \right)}
\newcommand{\ess} [3] {\mathrm{ess}^{#1}_{#2}\left( #3 \right)}
\newcommand{\rpt} [2] {\mathrm{rpt}_{#1}\left( #2 \right)}

\newcommand{\suppress}[1]{}
\newcommand{\etal}{\emph{et al.\/}}

\newcommand {\br} [1] {\ensuremath{ \left( #1 \right) }}
\newcommand {\Br} [1] {\ensuremath{ \left[ #1 \right] }}
\newcommand {\set} [1] {\ensuremath{ \left\lbrace #1 \right\rbrace }}
\newcommand {\minusspace} {\: \! \!}
\newcommand {\smallspace} {\: \!}
\newcommand {\fn} [2] {\ensuremath{ #1 \minusspace \br{ #2 } }}
\newcommand {\Fn} [2] {\ensuremath{ #1 \minusspace \Br{ #2 } }}
\newcommand {\eqdef} {\ensuremath{ \stackrel {\mathrm{def}} {=} }}
\newcommand {\fndec} [3] {\ensuremath{ #1 : \, #2 \rightarrow #3 }}
\newcommand {\mutinf} [2] {\fn{\mathrm{I}}{#1 \smallspace : \smallspace #2}}
\newcommand {\condmutinf} [3] {\mutinf{#1}{#2 \smallspace \middle\vert \smallspace #3}}
\newcommand {\prob} [1] {\Fn{\Pr}{#1}}
\newcommand {\pr} [2] {\Fn{\Pr_{#1}}{#2}}
\newcommand {\setfont} [1] {\ensuremath{\mathcal{#1}}}
\newcommand {\abs} [1] {\ensuremath{ \left| #1 \right| }}
\newcommand {\norm} [1] {\ensuremath{ \left\| #1 \right\| }}
\newcommand {\normsub} [2] {\ensuremath{ \norm{#1}_{#2} }}
\newcommand {\onenorm} [1] {\normsub{#1}{1}}
\newcommand {\SetToN} [1] {\set{ 1, 2, \ldots, #1 }}
\newcommand {\relent} [2] {\fn{\mathrm{S}}{#1 \middle\| #2}}
\newcommand {\rminent} [2] {\fn{\mathrm{S}_{\infty}}{#1 \middle\| #2}}
\newcommand {\expec} [2] {\Fn{\mathbb{E}_{\substack{#1}}}{#2}}
\newcommand {\modi}{\hl{MODIFIED}}
\newcommand {\mend}{\hl{END}}


%

\title{A strong direct product theorem in terms of  the smooth rectangle bound}
%
%
%
%
%

\numberofauthors{2} 
%
\author{
%
%
\alignauthor
Rahul Jain\\
       \affaddr{Centre for Quantum Technologies and Department of Computer Science}\\
       \affaddr{National University of Singapore}\\
       \email{rahul@comp.nus.edu.sg}
\alignauthor
Penghui Yao\\
       \affaddr{Centre for Quantum Technologies}\\
       \affaddr{National University of Singapore}\\
       \email{pyao@nus.edu.sg}
}

\maketitle
\begin{abstract}
A strong direct product theorem states that, in order to solve $k$ instances of a problem, if we provide
less than $k$ times the resource required to compute one instance, then the probability of overall success
is exponentially small in $k$. In this paper, we consider the model of two-way public-coin communication complexity and show a strong direct
product theorem for all relations in terms of the {\em smooth
rectangle bound}, introduced  by Jain and Klauck~\cite{JainKlauck2010} as a generic  lower bound method in this model. Our result therefore implies a strong direct product theorem for all relations for which an (asymptotically) optimal lower bound can be provided using the smooth rectangle bound. In fact we are not aware of any relation for which it is known that the smooth rectangle bound does not provide an optimal
lower bound. This lower bound subsumes many of the other known lower bound methods, for example the {\em rectangle bound} (a.k.a the {\em corruption bound})~\cite{Razborov92}, the {\em smooth discrepancy bound} (a.k.a  the {\em $\gamma_2$ bound}~\cite{Linial:2007} which in turn subsumes the {\em discrepancy bound}), the {\em subdistribution bound}~\cite{Jain2008} and the {\em conditional min-entropy bound}~\cite{Jain:2011}.

As a consequence, our result reproves some of the known strong direct product results, for example for \textsf{Inner Product}~\cite{Kushilevitz96}
and \textsf{Set-Disjointness}~\cite{Klauck2010,Jain:2011}. Our result also shows new strong direct product result for \textsf{Gap-Hamming Distance}~\cite{Chakrabarti:2011:OLB:1993636.1993644,Sherstov11a}  and also implies near optimal direct product results for several important functions and relations used to show exponential separations between classical and quantum communication complexity, for which
near optimal lower bounds are provided using the rectangle bound, for example by Raz~\cite{Raz:1999}, Gavinsky~\cite{Gavinsky:2008} and Klartag and
Regev~\cite{Regev:2011}. 

We show our result using information theoretic arguments. A key tool we use is a sampling protocol due to Braverman~\cite{Braverman2012},
in fact a modification of it used by Kerenidis, Laplante, Lerays, Roland and  Xiao~\cite{Kerenidis2012}.
\end{abstract}




\newpage

\section{Introduction}
Given a model of computation, suppose solving one instance of a given problem $f$ with probability of success
$p<1$ requires $c$ units of some resource. A natural question that may be asked is: how much
resource is needed to solve $f^k$, $k$ instances of the same problem, simultaneously. A naive way is by running the optimal protocol for $f$, $k$ times in parallel, which requires $c\cdot k$ units of resource, however the probability of overall success is $p^k$ (exponentially small in $k$). A {\it strong direct product conjecture} for $f$ states that this is essentially optimal, that is if only $o(k\cdot c)$ units of resource are provided  for any protocol solving $f^k$, then the probability of overall success is at most $p^{\Omega(k)}$.

Proving or disproving strong direct product conjectures in various models of computation has been a central task in theoretical computer science, notable examples of such results being Yao's {\em XOR lemma}~\cite{Yao82} and  Raz's
\cite{Raz:1995:PRT:225058.225181} theorem for two-prover games. Readers may refer to~\cite{Klauck2010, Jain:2011, Jain2012} for a good discussion of known results in different models of computation. In the present work, we consider the model of two-party two-way public-coin communication complexity~\cite{Yao:1979} and consider the direct product question in this model. In this model, there are two parties who wish to compute a joint function (more generally a relation) of their input, by doing local computation, sharing public coins and exchanging messages. The resource counted is the number of bits communicated between them.  The textbook by Kushilevitz and Nisan~\cite{Kushilevitz96} is an excellent reference for communication complexity.  Much effort has been made towards investigating direct product questions in this model and strong direct product theorems have been shown for many different functions, for example \textsf{Set-Disjointness}~\cite{Klauck2010,Jain:2011},  \textsf{Inner Product}~\cite{Lee2008}, \textsf{Pointer Chasing}~\cite{Jain2012} etc. To the best of our knowledge, it is not known if the strong direct product conjecture fails to hold for any function or relation in this model. Therefore, whether the strong direct product conjecture holds for all relations in this model, remains one of the major open problems in communication complexity. In the model of constant-round public-coin communication complexity, recently a strong direct product result has been shown to hold for all relations by Jain, Perezl\'{e}nyi and Yao~\cite{Jain2012}. The work~\cite{Jain2012} built on a previous result due to Jain~\cite{Jain:2011} showing  a strong direct product result for all relations in the model of one-way public-coin communication complexity (where a single message is sent from Alice to Bob, who then determines the answer).

The weaker {\it direct sum} conjecture, which states that solving $k$ independent instances of a problem with constant success probability requires $k$ times the resource needed to compute one instance
with constant success probability, has also been extensively investigated in different models of communication complexity and has met a better success. Direct sum theorems have been shown to hold for all relations in the public-coin one-way model~\cite{Jain:2003:DST:1759210.1759242}, the entanglement-assisted quantum one-way model~\cite{Jain2005}, the public-coin simultaneous message passing model~\cite{Jain:2003:DST:1759210.1759242}, the private-coin simultaneous message passing model~\cite{Jain:2009:NRS:1602931.1603157}, the constant-round public-coin two-way model~\cite{BravermanRao2011} and the model of two-way distributional communication complexity under product distributions~\cite{Barak:2010:CIC:1806689.1806701}. Again please refer to \cite{Klauck2010, Jain:2011, Jain2012} for a good discussion.

Another major focus in communication complexity has been to investigate generic lower bound methods, that apply to all functions (and possibly to all relations). In the model  we are concerned with, various generic lower bound methods are known, for example the {\em partition bound}~\cite{JainKlauck2010},  the {\em information complexity}~\cite{Chakrabarti01},  the {\em smooth rectangle bound}~\cite{JainKlauck2010} (which in turn subsumes the {\em rectangle bound} a.k.a the {\em corruption bound})~\cite{Yao:1983:LBP:1382437.1382849,Babai:1986:CCC:1382439.1382962,Razborov92,Klauck2003,Beame:2006:SDP:1269083.1269088}, the {\em smooth discrepancy bound} a.k.a  the {\em $\gamma_2$ bound}~\cite{Linial:2007} (which in turn subsumes the {\em discrepancy bound}), the {\em subdistribution bound}~\cite{Jain2008} and the {\em conditional min-entropy bound}~\cite{Jain:2011}. Proving strong direct product results in terms of these lower bound methods is a reasonable approach to attacking the general question.
Indeed, many lower bounds have been shown to satisfy strong direct product theorems, example the discrepancy bound~\cite{Lee2008}, the subdistribution bound under product distributions~\cite{Jain2008}, the smooth discrepancy bound~\cite{Sherstov:2011:SDP:1993636.1993643} and the conditional min-entropy bound~\cite{Jain:2011}.

\subsection*{Our result} In present work, we show a strong direct product theorem in terms of the smooth rectangle bound, introduced by Jain and Klauck~\cite{JainKlauck2010}, which generalizes the  rectangle bound (a.k.a. the corruption bound)
\cite{Yao:1983:LBP:1382437.1382849,Babai:1986:CCC:1382439.1382962,Razborov92,Klauck2003,Beame:2006:SDP:1269083.1269088}. Roughly speaking, the rectangle bound for relation $f\subseteq \X \times \Y \times \Z$ under a distribution $\mu$, with respect to an element $z\in \Z$, and error $\ve$, tries to capture the size (under $\mu$) of a largest rectangle for which $z$ is a right answer for $1 - \ve$ fraction of inputs inside the rectangle. It is not hard to argue that the rectangle bound forms a lower bound on the {\em distributional } communication complexity of $f$ under $\mu$. The smooth rectangle bound for $f$ further  captures the maximum, over all relations $g$ that are close to $f$  under $\mu$, of the rectangle bound  of $g$ under $\mu$. The distributional error setting can eventually be related to the worst case error setting via the well known Yao's principle~\cite{Yao:1983:LBP:1382437.1382849}.

Jain and Klauck showed that the smooth rectangle bound is stronger than every lower bound method we mentioned above except the partition bound and the information complexity. Jain and Klauck showed that the partition bound subsumes the smooth rectangle bound and in a recent work Kerenidis, Laplante, Lerays, Roland and  Xiao~\cite{Kerenidis2012} showed that the information complexity subsumes the smooth rectangle bound (building on the work of Braverman and Weinstein~\cite{BravermanWeinstein2011} who showed that the information complexity subsumes the discrepancy bound). New lower bounds for specific functions have been discovered using the smooth rectangle bound, for example Chakrabarti and Regev's~\cite{Chakrabarti:2011:OLB:1993636.1993644} optimal lower bound for the \textsf{Gap-Hamming Distance} partial function. Klauck~\cite{Klauck2010} used the smooth rectangle bound to show a strong direct product result for the \textsf{Set-Disjointness} function, via exhibiting a lower bound on a related function. On the other hand, as far as we know, no function (or relation) is known for which its  smooth rectangle bound is (asymptotically) strictly smaller than its two-way public-coin communication complexity. Hence establishing whether or not the smooth rectangle bound is a tight lower bound for all functions and relations in this model is an important open question.
Our result is as follows.

\begin{theorem}\footnote{For a relation $f$ with $\mathrm{R}^{\mathrm{pub}}_{\ve}(f) = \O(1)$, a strong direct product result can be shown via direct arguments~\cite{Jain:2011}.}
\label{thm:sdptsrec}
Let $\X,\Y,\Z$ be finite sets, $f\subseteq\X\times\Y\times\Z$ be a relation, and $t>1$ be an integer. Let $\mu$ be a distribution on $\X\times\Y$. Let $z \in \Z$ and $\beta\defeq\pr{(x,y)\leftarrow \mu}{f(x,y)=\set{z}}$.  Let $0<\ve<1/3$ and $\ve',\delta>0$ be such that $\frac{\delta+22\ve}{\beta-33\ve} < (1+\ve')\frac{\delta}{\beta}$. It holds that,
\[\mathrm{R}^{\mathrm{pub}}_{1-(1-\ve)^{\lfloor\ve^2 t/32\rfloor}}(f^t)\geq\frac{\ve^2}{32} \cdot t \cdot \br{ 11 \ve \cdot \lsrec{z,\mu}{(1+\ve')\delta/\beta,\delta}{f}  - 2}.\]
\end{theorem}
Above $\mathrm{R}^{\mathrm{pub}}(\cdot)$ represents the two-way public-coin communication communication complexity and $\lsrec{}{}{\cdot}$ represents the smooth rectangle bound (please refer to Section~\ref{sec:preliminary} for precise definitions). 
Our result  implies a strong direct product theorem for all relations for which an (asymptotically) optimal lower bound can be provided using the smooth rectangle bound.  As a consequence, our result reproves some of the known strong direct product results, for example for \textsf{Inner Product}~\cite{Kushilevitz96} 
and \textsf{Set-Disjointness}~\cite{Klauck2010,Jain:2011}. Our result also shows new strong direct product result for \textsf{Gap-Hamming Distance}~\cite{Chakrabarti:2011:OLB:1993636.1993644,Sherstov11a}  and also implies near optimal direct product results for several important functions and relations used to show exponential separations between classical and quantum communication complexity, for which
near optimal lower bounds are provided using the rectangle bound, for example by Raz~\cite{Raz:1999}, Gavinsky~\cite{Gavinsky:2008} and Klartag and
Regev~\cite{Regev:2011}. 

In a recent work, Harsha and Jain~\cite{HarshaJ2013} have shown that the smooth-rectangle bound provides an optimal lower bound of $\Omega(n)$ for the $\mathsf{Tribes}$ function. For this function the rectangle bound fails to provide an optimal lower bound since it is $O(\sqrt{n})$. Earlier Jayram, Kumar and Sivakumar~\cite{JayramKS2003} had shown a lower bound of $\Omega(n)$ using {\em information~complexity}. The result of~\cite{HarshaJ2013} along with Theorem~\ref{thm:sdptsrec} (and Lemma~\ref{lem:eqv} appearing in the Appendix, which relates two different definitions of the smooth-rectangle bound) implies a strong direct product result for the  $\mathsf{Tribes}$ function.
 
In \cite{Kerenidis2012}, Kerenidis et. al. introduced the {\it relaxed partition bound} (a weaker version of the partition bound~\cite{JainKlauck2010}) and showed it to be stronger than the smooth rectangle bound. It is easily seen (by comparing the corresponding linear-programs) that the smooth rectangle bound and the relaxed partition bound are in-fact equivalent for boolean functions (and more generally when the size of output set is a constant). Thus our result also implies
a strong direct product theorem in terms of the relaxed partition bound for boolean functions (and more generally when the size of output set is a constant).


\subsection*{Our techniques}

The broad argument of the proof of our result is as follows. We show our result in the  distributional error setting and translate it to the worst case error setting using the well known Yao's principle \cite{Yao:1983:LBP:1382437.1382849}. Let $f$ be a relation,  $\mu$ be a distribution on $\X\times\Y$, and $c$ be the smooth rectangle bound of $f$ under the distribution $\mu$ with output $z \in \Z$. Consider a protocol $\Pi$ which computes $f^k$ with inputs drawn from distribution $\mu^k$ and communication $o(c \cdot k)$ bits. Let $\C$ be a subset of the coordinates $\{1, 2, \ldots, k\}$. If the probability that $\Pi$ computes all the instances in $\C$ correctly is as small as desired, then we are done. Otherwise, we exhibit a new coordinate $j\notin \C$, such  that the probability, conditioned on success in $\C$,  of the protocol $\Pi$ answering correctly in the $j$-th coordinate is bounded away from $1$. Since $\mu$ could be a non-product distribution we introduce a new random variable $R_j$, such that conditioned on it and $X_jY_j$ (input in the $j$th coordinate), Alice and Bob's inputs in the other coordinates become independent. Use of such a variable to handle non product distributions has been used in many previous works, for example~$\cite{Bar-Yossef2002a,Holenstein2007,Barak:2010:CIC:1806689.1806701,Jain:2011,Jain2012}$.

Let  the random variables $X^1_jY^1_jR_j^1M^1$  represent the inputs  in the $j$th coordinate, the new variable $R_j$ and the
message transcript of $\Pi$, conditioned on the success on $\C$. The first useful property that we observe is that the joint distribution of
$X^1_jY^1_jR_j^1M^1$ can be written as,
$$\pr{}{X^1_jY^1_jR^1_jM^1=xym}=\frac{1}{q}\mu(x,y)u_x(r_j,m)u_y(r_j,m), $$
where $u_x,u_y$ are functions and $q$ is a positive real number.
The marginal distribution of $X_j^1Y_j^1$ is no longer $\mu$ though.  However (using arguments as in~\cite{Jain:2011,Jain2012}),
one can show that the distribution of $X_j^1Y_j^1$ is close, in $\ell_1$ distance, to $\mu$ and
$\condmutinf{X^1_j}{R^1_jM^1}{Y^1_j}+\condmutinf{Y^1_j}{R^1_jM^1}{X^1_j}\leq o(c)$, where $\mutinf{}{}$ represents the mutual information (please refer to Section~\ref{sec:preliminary} for precise definitions) .

Now, assume for contradiction that the success in the $j$th coordinate in $\Pi$ is large, like $0.99$, conditioned on success in $\C$.  Using the conditions obtained in the previous paragraph, we argue that there exists a zero-communication public-coin protocol $\Pi'$,  between Alice and Bob, with inputs drawn from $\mu$. In $\Pi'$ Alice and Bob are allowed to abort the protocol or output an element in $\Z$. We show that the probability of non-abort for this protocol is large, like $2^{-c}$, and conditioned on non-abort, the probability that Alice and Bob output a correct answer for their inputs is also large, like $0.99$. This allows us to exhibit (by fixing the public coins of $\Pi'$ appropriately),   a large rectangle (with weight under $\mu$ like $2^{-c}$) such that $z$ is a correct answer for a large fraction (like $0.99$) of the inputs inside the rectangle. This shows that the rectangle bound of $f$, under $\mu$ with output $z$, is smaller than $c$. With careful analysis we are also able to show that the smooth rectangle bound of $f$ under $\mu$, with output $z$, is smaller than $c$, reaching a contradiction to the definition of $c$.

The sampling protocol that we use to obtain the public-coin zero communication protocol, is the same as that in Kerenidis et al.~\cite{Kerenidis2012}, which in turn is a modification of a protocol due to Braverman~\cite{Braverman2012}\footnote{A protocol, achieving similar task, however working only for product distributions on inputs was first shown by Jain, Radhakrishnan and Sen~\cite{Jain2005}.} (a variation of which also appears in~\cite{BravermanWeinstein2011}). However our analysis of the protocol's correctness deviates significantly in parts from the earlier works~\cite{Kerenidis2012,Braverman2012,BravermanWeinstein2011} due to the fact that for us the marginal distribution of $X^1Y^1$ need not be the same as that of $\mu$, in fact for some inputs $(x,y)$, the probability under the two distributions can be significantly different. 

There is another important original contribution of our work, not present in the previous works~\cite{Kerenidis2012,Braverman2012,BravermanWeinstein2011}. We observe a crucial property of the protocol $\Pi'$ which turns out to be very important in our arguments. The property is that the bad inputs $(x,y)$ for which the distribution of $\Pi'$'s sample for $R^1_jM^1$, conditioned on non-abort, deviates a lot from the desired  $R^1_jM^1 |~(X^1Y^1= xy)$, their probability is nicely reduced (as compared to $\pr{}{X^1Y^1=xy}$) in the final distribution of $\Pi'$, conditioned on non-abort.  This helps us to argue that the distribution of inputs and outputs in $\Pi'$, conditioned on non-abort, is close in $\ell_1$ distance to $X^1_jY^1_jR^1_jM^1$, implying good success in $\Pi'$, conditioned on non-abort.

\vspace{0.2cm}

\noindent {\it Organization}.  In Section \ref{sec:preliminary}, we present some necessary background, definitions and preliminaries. In Section \ref{sec:proof}, we prove our main result Theorem \ref{thm:sdptsrec}.  We defer some proofs to Appendix due to lack of space.


\section{Preliminary}
\label{sec:preliminary}

\subsection*{Information theory}

We use capital letters  e.g. $X,Y,Z$ or letters in bold e.g. $\mathbf{a},\mathbf{b}, \mbox{\boldmath$\alpha$},\mbox{\boldmath$\beta$}$
 to represent random variables and use calligraphic letters e.g. $\X,\Y,\Z$ to represent sets.
For integer $n \geq 1$, let $[n]$ represent the set $\{1,2, \ldots,
n\}$. Let $\X$, $\Y$ be finite sets and $k$ be a natural number. Let
$\X^k$ be the set $\X\times\cdots\times\X$, the cross product of
$\X$, $k$ times. Let $\mu$ be a (probability) distribution on $\X$.
Let $\mu(x)$ represent the probability of $x\in\X$ according to
$\mu$. For any subset $S\subseteq \X$, define $\mu(S)\defeq\sum_{x\in S}\mu(x)$.
 Let $X$ be a random variable distributed according to $\mu$,
which we denote by $X\sim\mu$. We use the same symbol to represent
a random variable and its distribution whenever it is clear from
the context.  The expectation value of  function $f$ on $\X$ is
denoted as
$\expec{x \leftarrow X}{f(x)} \defeq\sum_{x\in\X} \prob{X=x} \cdot f(x). $
The entropy of $X$ is defined as
$\mathrm{H}(X)\defeq-\sum_x\mu(x) \cdot \log\mu(x)$ ($\log, \ln $ represent logarithm to the base $2, e$ repectively). For two distributions
$\mu$, $\lambda$ on $\X$, the distribution $\mu \otimes \lambda$ is
defined as $(\mu\otimes\lambda)(x_1,x_2)\defeq\mu(x_1)\cdot\lambda(x_2)$.
Let $\mu^k\defeq\mu\otimes\cdots\otimes\mu$, $k$ times. The $\ell_1$
distance between $\mu$ and $\lambda$ is defined to be half of the
$\ell_1$ norm of $\mu - \lambda$; that is,
$\|\lambda-\mu\|_1\defeq\frac{1}{2}\sum_x|\lambda(x)-\mu(x)|=\max_{S\subseteq\X}|\lambda(S)-\mu(S)|$. We say that
$\lambda$ is $\ve$-close to $\mu$ if $\|\lambda-\mu\|_1\leq\ve$. The
relative entropy between distributions $X$ and $Y$ on $\X$ is
defined as
 $\relent{X}{Y} \defeq \expec{x\leftarrow X}{\log \frac{\prob{X=x}}{\prob{Y=x}}} .$
The relative min-entropy between them is defined as
$ \rminent{X}{Y} \defeq \max_{x\in\X}
    \set{ \log \frac{\prob{X=x}}{\prob{Y=x}} }.$
It is easy to see that $\relent{X}{Y} \leq \rminent{X}{Y}$. Let
$X,Y,Z$ be jointly distributed random variables. Let $Y_x$ denote the
distribution of $Y$ conditioned on $X=x$. The conditional entropy of
$Y$ conditioned on $X$ is defined as $\mathrm{H}(Y|X) \defeq
\expec{x\leftarrow X}{\mathrm{H}(Y_x)} =
\mathrm{H}(XY)-\mathrm{H}(X)$. The mutual information between $X$
and $Y$ is defined as:
$  \mutinf{X}{Y} \defeq \mathrm{H}(X)+\mathrm{H}(Y)-\mathrm{H}(XY)
    = \expec{y \leftarrow Y}{\relent{X_y}{X}}
    = \expec{x \leftarrow X}{\relent{Y_x}{Y}}. $
The conditional mutual information between $X$ and $Y$, conditioned on
$Z$, is defined as:
$  \condmutinf{X}{Y}{Z} \defeq
    \expec{z \leftarrow Z}{\condmutinf{X}{Y}{Z=z}}
    = \mathrm{H}\br{X|Z}+\mathrm{H}\br{Y|Z}-\mathrm{H}\br{XY|Z} .$
The following {\em chain rule} for mutual information is easily
seen :
$\mutinf{X}{YZ} = \mutinf{X}{Z} + \condmutinf{X}{Y}{Z} .$

We will need the following basic facts. A very good text for
reference on information theory is~\cite{CoverT91}.
\begin{fact}
\label{fact:relative entropy joint convexity} Relative entropy is
jointly convex in its arguments. That is, for distributions $\mu,
\mu^1, \lambda, \lambda^1 \in \X$ and $p\in[0,1]$:
$ \relent{p \mu  + (1-p) \mu^1}{\lambda + (1-p) \lambda^1} \leq p \cdot \relent{\mu}{\lambda} + (1-p) \cdot \relent{\mu^1}{\lambda^1} .$
\end{fact}
\begin{fact}
    \label{fact:relative entropy splitting}
Relative entropy satisfies the following chain rule. Let $XY$ and
$X^1Y^1$ be random variables on $\X\times\Y$.
    It holds that:
    $ \relent{X^1Y^1}{XY} = \relent{X^1}{X}
    + \expec{x\leftarrow X^1} {\relent{Y^1_x}{Y_x}}.$
    In particular,
    $ \relent{X^1Y^1}{X\otimes Y}
    = \relent{X^1}{X} + \expec{x\leftarrow X^1}{\relent{Y^1_x}{Y}}
    \geq \relent{X^1}{X} + \relent{Y^1}{Y}.$ The last inequality follows from Fact~\ref{fact:relative entropy joint convexity}.
\end{fact}
\begin{fact} \label{fact:mutinf is min}
    Let $XY$ and $X^1Y^1$ be random variables on $\X\times\Y$.
    It holds that
    \[  \relent{X^1Y^1}{X\otimes Y}
    \geq \relent{X^1Y^1}{X^1\otimes Y^1}=\mutinf{X^1}{Y^1}. \]
\end{fact}


The following fact follows from  Fact \ref{fact:relative entropy splitting} and Fact \ref{fact:mutinf is min}.
\begin{fact}
   \label{fact:12}
   Given random variables $XY$ and $X'Y'$ on $\X\times\Y$, it holds that
   $$\expec{x\leftarrow X'}{\relent{Y'_x}{Y}}\geq\expec{x\leftarrow X'}{\relent{Y'_x}{Y'}}=\mutinf{X'}{Y'}.$$
\end{fact}

\begin{fact}
    \label{fact:one norm and rel ent}
    For distributions $\lambda$ and $\mu$: \quad $0 \leq \onenorm{\lambda-\mu} \leq \sqrt{\relent{\lambda}{\mu}}$.
\end{fact}

\begin{fact}(\textbf{Classical substate theorem}~\cite{Jain2002})\label{fact:markovofrelent}
   Let $X,X'$ be two distributions on $\X$. For any $\delta\in(0,1)$, it holds that
   \[\pr{x\leftarrow X'}{\frac{\pr{}{X'=x}}{\pr{}{X=x}}\leq 2^{\left(\relent{X'}{X}+1\right)/\delta}}\geq 1-\delta.\]
\end{fact}

%

%
%

We will need the following lemma. Its proof is deferred to Appendix.
\begin{lemma}
   \label{lem:ratiovs1}
   Given random variables $A$, $A'$ and $\ve>0$, if $\onenorm{A-A'}\leq\ve$, then for any $r\in(0,1)$,
   \begin{eqnarray*}
   &\pr{a\leftarrow A}{\abs{1-\frac{\pr{}{A'=a}}{\pr{}{A=a}}}\leq\frac{\ve}{r}}\geq1-2r; ~\text{and}~ \\
   &\pr{a\leftarrow A'}{\abs{1-\frac{\pr{}{A'=a}}{\pr{}{A=a}}}\leq\frac{\ve}{r}}\geq1-2r-\ve.
   \end{eqnarray*}
\end{lemma}


\subsection*{Communication complexity}

Let $\X,\Y,\Z$ be finite sets, $f\subseteq\X\times\Y\times\Z$ be a relation and $\ve >0$. In a two-way public-coin communication protocol,
Alice is given $x\in\X$, and Bob is given $y\in\Y$. They are supposed to output $z\in\Z$ such that $(x,y,z)\in f$ via exchanging messages and doing local computations. They may share public coins before the inputs are revealed to them. We assume that the last $\lceil\log|\Z|\rceil$ bits of the transcript is the output of the protocol. Let $\mathrm{R}^{\text{pub}}_{\ve}(f)$ represent
the two-way public-coin randomized communication complexity of $f$ with the worst case error $\ve$, that is the communication of the best two-way public-coin protocol for $f$ with error for each input $(x,y)$ being at most $\ve$. Let $\mu$ be a distribution on $\X\times\Y$. Let $\mathrm{D}^{\mu}_{\ve}(f)$ represent the two-way distributional communication complexity of $f$ under distribution $\mu$ with distributional error $\ve$, that is the communication of the best
two-way deterministic protocol for $f$, with average error over the distribution of the inputs drawn from $\mu$, at most $\ve$. Following is Yao's min-max principle which connects the worst case error and the distributional error settings, see. e.g., ~\cite[Theorem~3.20, page~36]{Kushilevitz96}.
\begin{fact}\label{fact:yaos principle}\cite{Yao:1979}
$\mathrm{R}^{\mathrm{pub}}_{\ve}(f)=\max_{\mu}\mathrm{D}^{\mu}_{\ve}(f)$.
\end{fact}
The following fact can be easily verified by induction on the number of message exchanges in a private-coin protocol (please refer for example to~\cite{Braverman2012} for an explicit proof). It is also implicit in the {\em cut and paste} property of private-coins protocol used in Bar-Yossef, Jayram, Kumar and Sivakumar~\cite{Bar-Yossef2002a}.
\begin{lemma} \label{lem:privateccdistribution}
For any private-coin two-way communication protocol, with input $XY\sim\mu$ and transcript $M \in \M$, the joint distribution can be written as
\[\pr{}{XYM=xym}=\mu(x,y)u_x(m)u_y(m),\]
where $u_x:\M\rightarrow[0,1]$ and $u_y:\M\rightarrow[0,1]$, for all $(x,y)\in\X\times\Y$.
\end{lemma}

\subsection*{Smooth rectangle bound}

Let $f\subseteq\X\times\Y\times\Z$ be a relation and $\ve,\delta \geq 0$. With a slight abuse of notation, we write $f(x,y)\defeq\set{z\in\Z | ~(x,y,z)\in f}$, and $f^{-1}(z)\defeq\set{(x,y):(x,y,z)\in f}$.

\begin{definition}\label{def:srec}(\textbf{Smooth-rectangle bound}~\cite{JainKlauck2010})
The $(\ve,\delta)$-smooth rectangle bound of $f$, denoted by
$\lsrec{}{\ve,\delta}{f}$, is defined as follows:
\begin{eqnarray*}
&&\lsrec{}{\ve,\delta}{f}\defeq\max\{\lsrec{\lambda}{\ve,\delta}{f}| ~\lambda\ \text{a distribution over}\ \X\times\Y\};\\
&&\lsrec{\lambda}{\ve,\delta}{f}\defeq\max\{\lsrec{z,\lambda}{\ve,\delta}{f}|~z\in\Z\};\\
&&\lsrec{z,\lambda}{\ve,\delta}{f}\defeq\max\{\lrec{z,\lambda}{\ve}{g}|~g\subseteq\X\times\Y\times\Z;\\
&&\pr{(x,y)\leftarrow\lambda}{f(x,y)\neq g(x,y)}\leq\delta\};\\
&&\lrec{z,\lambda}{\ve}{g}\defeq\min\{\rminent{\lambda_R}{\lambda}| ~R\ \text{is a rectangle in $\X\times\Y$},\\
&&\lambda(g^{-1}(z)\cap R)\geq(1-\ve)\lambda(R)\}.
\end{eqnarray*}

\end{definition}

When $\delta=0$, the smooth rectangle bound equals the rectangle bound (a.k.a. the corruption bound)
\cite{Yao:1983:LBP:1382437.1382849,Babai:1986:CCC:1382439.1382962,Razborov92,Klauck2003,Beame:2006:SDP:1269083.1269088}. Definition \ref{def:srec} is a generalization of the one in \cite{JainKlauck2010}, where it is only defined
for boolean functions. The smooth rectangle bound is a lower bound on the two-way public-coin communication complexity. The proof of the following lemma appears in Appendix.
\begin{lemma}
\label{lem:Dgeqsrec}
Let  $f\subseteq\X\times\Y\times\Z$ be a relation.  Let $\lambda \in \X \times \Y$ be a distribution and let $z \in \Z$.  Let 
$\beta \defeq \pr{(x,y)\leftarrow\lambda}{f(x,y)=\set{z}}$.  Let $\ve, \ve', \delta >0$ be such that $\frac{\delta+\ve}{\beta-2\ve} < (1+\ve') \frac{\delta}{\beta}.$ Then, 
$$ \mathrm{R}_{\ve}(f) \geq \mathrm{D}^{\lambda}_{\ve}(f) \geq \lsrec{z,\lambda}{(1+\ve')\delta/\beta,\delta} {f} - \log \frac{4}{\ve} .$$
\end{lemma}



\section{Proof}\label{sec:proof}

 The following lemma builds a connection between the zero-communication protocols and the smooth rectangle bound.

\begin{lemma}
  \label{lem:zeroprotocolimpliesrec}
  Let $f\subseteq\X\times\Y\times\Z$, $X'Y' \in \X \times \Y$ be a distribution and $z \in \Z$. Let $\beta \defeq \pr{(x,y)\leftarrow X'Y'}{f(x,y)=\set{z}}$.  Let $c\geq 1$. Let $\ve, \ve', \delta > 0$ be such that $(\delta+2\ve)/(\beta-3\ve) < (1+\ve')\delta/\beta$.  Let $\Pi$ be a  zero-communication public-coin protocol with input $X'Y'$,  public coin $R$, Alice's output $A\in\Z\cup\set{\bot}$, and Bob's output
  $B\in\Z\cup\set{\bot}$. Let  $X^1Y^1A^1B^1R^1 \defeq (X'Y'ABR | ~A = B \neq\bot)$. Let
  \begin{enumerate}

  \item[1.] $\pr{}{A = B \neq\bot}\geq 2^{-c}$ ; \qquad   2. $\norm{X^1Y^1-X'Y'}\leq\ve$.

  \item[3.] $\pr{}{(X^1, Y^1,  A^1) \in f }\geq1-\ve$.

  \end{enumerate}

\noindent Then $\lsrec{z,X'Y'}{(1+\ve')\delta/\beta,\delta}{f}<\frac{c}{\ve}\enspace .$
\end{lemma}

\begin{proof}
 Let $g\subseteq\X\times\Y\times\Z$, satisfy 
$$\pr{(x,y)\leftarrow X'Y'}{f(x,y)\neq g(x,y)}\leq\delta.$$ 
  It suffices to show that   $\lrec{z, X'Y'}{(1+\ve')\delta/\beta}{g}\leq \frac{c}{\ve}$. \newline
Since $\pr{}{A = B \neq\bot}\geq 2^{-c}$,
   \begin{align}
   c & \geq \rminent{X^1Y^1R^1A^1B^1}{X'Y'RAB}  \\
& \geq \relent{X^1Y^1R^1A^1B^1}{X'Y'RAB}  \nonumber \\
 & \geq \expec{r\leftarrow R^1, a \leftarrow A^1}{\relent{(X^1Y^1)_{r,a}}{X'Y'}} \quad \mbox{(from Fact~\ref{fact:relative entropy splitting})}. \label{eqn:zeroccsize}
   \end{align}
Since $\norm{X^1Y^1-X'Y'}\leq\ve$,
  \begin{align}
   \pr{xyr\leftarrow X^1Y^1R^1}{f(x,y) = \set{z}} & \geq \pr{xy\leftarrow X'Y'}{f(x,y) = \set{z}} - \ve  \nonumber \\
 & \geq \beta - \ve.  \label{eqn:zeroccess}
  \end{align}
Since  $\pr{}{(X^1, Y^1,  A^1) \in f }\geq1-\ve$, hence $\prob{A^1=B^1=z} \geq \beta - 2\ve$.
Since $$\pr{(x,y)\leftarrow X'Y'}{f(x,y)\neq g(x,y)}\leq\delta,$$ by item 2 of this lemma, we have
  \begin{eqnarray}
  &\pr{xyra\leftarrow X^1Y^1R^1A^1}{(x,y,a)\in g} \geq  \nonumber\\
  &\pr{xyra\leftarrow X^1Y^1R^1A^1}{(x,y,a)\in f} - \delta -\ve \geq 1-2\ve - \delta. \label{eqn:zeroccerr}
  \end{eqnarray}
  By standard application of Markov's inequality on equations~\eqref{eqn:zeroccsize}, \eqref{eqn:zeroccess}, \eqref{eqn:zeroccerr}, we get an $r_0$, such that
  \begin{eqnarray*}
\relent{(X^1Y^1)_{r_0,z}}{X'Y'}&\leq&\frac{c}{\ve} , \label{eqn:zeroccrrrrsize}\\
\pr{xy\leftarrow (X^1Y^1)_{r_0,z}}{g(x,y) \neq \{z\}}&\leq& (\delta+2\ve)/(\beta-3\ve) \\
&\leq& (1+\ve')\delta/\beta .\label{eqn:zeroccrrrress}
  \end{eqnarray*}
 Here, $(X^1Y^1)_{r_0,z} = (X^1Y^1|(R^1=r_0, A^1=z)$. Note that the distribution of $(X^1Y^1)_{r_0,z}$ is  the distribution of $X'Y'$ restricted to some rectangle and then rescaled to make a distribution. Hence $$\relent{(X^1Y^1)_{r_0,z}}{X'Y'}=\rminent{(X^1Y^1)_{r_0,z}}{X'Y'}.$$
Thus
  $\lrec{z, X'Y'}{(1+\ve')\delta/\beta}{g}< \frac{c}{\ve}$.
\end{proof}

The following is our main lemma. A key tool that we use here is a sampling protocol that appears in~\cite{Kerenidis2012}  (protocol $\Pi'$ as
shown in Figure \ref{fig:pro}), which is a variant of a sampling protocol that appears in~\cite{BravermanWeinstein2011}, which in turn is a variant
of a sampling protocol that appears in~\cite{Braverman2012}.  Naturally similar arguments and calculations, as in this lemma, are made in previous
works~\cite{Braverman2012, BravermanWeinstein2011,Kerenidis2012}, however with a key difference. In their setting $\sum_{m}u_x(m)u_y(m) = 1$ for
all $(x,y)$. However in our setting this number could be much smaller than one for different $(x,y)$. Hence our arguments and
calculations deviate from  previous works at several places significantly.

\begin{lemma}(\textbf{Main Lemma}) \label{lem:mainlemma} Let $c\geq 1$. Let $p$ be a distribution over $\X\times\Y$ and $z \in \Z$.  Let $\beta\defeq\pr{(x,y)\leftarrow p}{f(x,y)=\set{z}}$. Let $0<\ve<1/3$ and $\delta, \ve' >0$ be such that $\frac{\delta+22\ve}{\beta-33\ve} < (1+\ve')\frac{\delta}{\beta}$. Let $XYM$ be random variables jointly distributed over the set $\X\times\Y\times\M$ such that the last $\lceil\log|\Z|\rceil$ bits of $M$ represents an element in $\Z$. Let $u_x:\M\rightarrow[0,1]$, 
$u_y:\M\rightarrow[0,1]$ be functions for all $(x,y)\in\X\times\Y$. If it holds that,
\begin{enumerate}
   \item For all $(x,y,m)\in\X\times\Y\times\M$, 
$$\pr{}{XYM=xym}=\frac{1}{q}p(x,y)u_x(m)u_y(m),$$  
where $q\defeq\sum_{xym}p(x,y)u_x(m)u_y(m)$;
   \item $\relent{XY}{p}\leq\ve^2/4$;
   \item $\condmutinf{X}{M}{Y}+\condmutinf{Y}{M}{X}\leq c$;
   \item $\err{f}{XYM}\leq\ve$, where $$\err{f}{XYM}\defeq\pr{xym\leftarrow XYM}{(x, y, \tilde{m})\notin f},$$
   and  $\tilde{m}$ represents the last $\lceil\log|\Z|\rceil$ bits of $m$;
\end{enumerate}
\noindent then 
$$\lsrec{z,p}{(1+\ve')\delta/\beta,\delta}{f} < \frac{2c}{11\ve^3}. \qed$$
\end{lemma}

\begin{proof}
Note by direct calculations,
\begin{eqnarray}
&&\pr{}{XY=xy}=\frac{1}{q}p(x,y)\alpha_{xy}, \label{eqn:xy}\\
&& \hspace{1in} \text{where $\alpha_{xy}\defeq\sum_{m}u_x(m)u_y(m)$};\nonumber\\
&&\pr{}{X=x}=\frac{1}{q}p(x)\alpha_x, \text{where $\alpha_x\defeq\sum_yp(y|x)\alpha_{xy}$} \label{eqn:x}\\
&&\pr{}{Y=y}=\frac{1}{q}p(y)\alpha_y, \text{where $\alpha_y\defeq\sum_xp(x|y)\alpha_{xy}$} \label{eqn:y}\\
&&\pr{}{X_y=x}=\frac{p(x|y)\alpha_{xy}}{\alpha_x},  \pr{}{Y_x=y}=\frac{p(y|x)\alpha_{xy}}{\alpha_y}\label{eqn:y_x}\\
&&\pr{}{M_{xy}=m}=u_x(m)u_y(m)/\alpha_{xy}; \label{eqn:mxy} \\
&&\pr{}{M_x=m}=\frac{u_x(m)v_x(m)}{\alpha_x} \label{eqn:mx}\\
&& \hspace{1in}  \mbox{where $v_x(m)\defeq\sum_yp(y|x)u_y(m)$};\nonumber \\
&&\pr{}{M_y=m}=\frac{u_y(m)v_y(m)}{\alpha_y}, \label{eqn:my}
\end{eqnarray}
\begin{eqnarray}
&& \hspace{1in}  \mbox{where $v_y(m)\defeq\sum_xp(x|y)u_x(m)$}.\nonumber
\end{eqnarray}

\begin{figure}[!ht]

\noindent\hrulefill

\noindent {Alice's input is $x$. Bob's input is $y$. Common input is $c, \ve, q, \M$.}

\medskip

\begin{enumerate}


\item Alice and Bob both set $\Delta\defeq\frac{c/\ve+1}{\ve}+2, T\defeq\frac{2}{q}|\M|2^{\Delta}\ln\frac{1}{\ve} \mbox{ and }  k\defeq\log(\frac{3}{\ve}(\ln\frac{1}{\ve}))$.

\item For $i=1,\cdots, T$ :

\begin{enumerate}

\item Alice and Bob, using public coins, jointly sample $\mathbf{m_i}\leftarrow\M, \mbox{\boldmath$\alpha_i$},\mbox{\boldmath$\beta_i$}\leftarrow[0,2^{\Delta}]$, uniformly.

\item Alice accepts $\mathbf{m_i}$ if $\mbox{\boldmath$\alpha_i$}\leq u_x(\mathbf{m_i}),$ and $\mbox{\boldmath$\beta_i$}\leq2^{\Delta}v_x(\mathbf{m_i})$.

\item Bob accepts $\mathbf{m_i}$ if $\mbox{\boldmath$\alpha_i$}\leq2^{\Delta}v_y(\mathbf{m_i}),$ and $\mbox{\boldmath$\beta_i$}\leq u_y(\mathbf{m_i})$.


\end{enumerate}

\item Let $\A\defeq\set{i\in[T]: \text{Alice accepts $\mathbf{m_i}$}}$ and $\B\defeq\set{i\in[T]: \text{Bob accepts $\mathbf{m_i}$}}$.

\item Alice and Bob, using public coins, choose a uniformly random function $\mathbf{h}:\M\rightarrow\{0,1\}^k$ and a uniformly random string $\mathbf{r}\in\set{0,1}^k$.

\begin{enumerate}

  \item Alice  outputs $\bot$ if either $\A$ is empty or $\mathbf{h}(\mathbf{m_i}) \neq \mathbf{r}$ (where $i$ is the smallest element in non-empty $\A$).  Otherwise,
  she outputs the element in $\Z$, represented by the last $\lceil\log|\Z|\rceil$ bits of $\mathbf{m_i}$.

  \item Bob finds the smallest $j\in\B$ such that $\mathbf{h}(\mathbf{m_j})=\mathbf{r}$. If no such $j$ exists, he outputs $\bot$.
  Otherwise, he outputs the element in $\Z$, represented by the last $\lceil\log|\Z|\rceil$ bits of $\mathbf{m_j}$.

\end{enumerate}

\end{enumerate}

\noindent\hrulefill

\caption{Protocol $\Pi'$ } \label{fig:pro}

\end{figure}

Define
\begin{eqnarray}
&&G_1\defeq\{(x,y):\abs{1-\frac{\alpha_{xy}}{q}}\leq \frac{1}{2}~\text{and}~\abs{1-\frac{\alpha_{x}}{q}}\leq \frac{1}{2}~\text{and}~\nonumber\\
&&\quad \quad\quad \abs{1-\frac{\alpha_{xy}}{q}}\leq \frac{1}{2}\};\label{eqn:g1}\\
&&G_2\defeq\set{(x,y):\relent{M_{xy}}{M_x}+\relent{M_{xy}}{M_y}\leq c/\ve};\label{eqn:g2}\\
&&G\defeq\{(x,y):\pr{m\leftarrow M_{xy}}{\frac{u_y(m)}{v_x(m)}\leq2^{\Delta}~\text{and}~ \frac{u_x(m)}{v_y(m)}\leq2^{\Delta}}\nonumber\\
&&\quad\quad\quad\geq1-2\ve\}.\label{eqn:g}
\end{eqnarray}
We begin by showing that $G_1 \cap G_2$ is a large set and also $G_1 \cap G_2 \subseteq G$.
\begin{claim}\label{claim:probofg}
\begin{enumerate}
\item $\pr{(x,y)\leftarrow p}{(x,y)\in G_1}>1-6\ve$,
\item $\pr{(x,y)\leftarrow p}{(x,y)\in G_2}\geq1-3\ve/2$,
\item $\pr{(x,y)\leftarrow p}{(x,y)\in G_1 \cap G_2}\geq1-15\ve/2,$
\item $G_1 \cap G_2 \subseteq G$.
\end{enumerate}
\end{claim}

\begin{proof}

Note item 1. and item 2. imply item 3. Now we show 1. Note that (using item 2. of Lemma~\ref{lem:mainlemma} and Fact~\ref{fact:one norm and rel ent}) $\onenorm{XY-p}\leq\ve/2$. From Lemma \ref{lem:ratiovs1} and (\ref{eqn:xy}), we have
\[\pr{(x,y)\leftarrow p}{\abs{1-\frac{\alpha_{xy}}{q}}\leq1/2}\geq1-2\ve.\]
By the monotonicity of $\ell_1$-norm, we have $\onenorm{X-p_{\X}}\leq\frac{\ve}{2}$ and $\onenorm{X-p_{\Y}}\leq\frac{\ve}{2}$. Similarly, from (\ref{eqn:x}) and (\ref{eqn:y}) we have
\begin{eqnarray*}
&&\pr{(x,y)\leftarrow p}{\abs{1-\frac{\alpha_x}{q}}\leq1/2}\geq1-2\ve, \quad \text{and} \\
&&\pr{(x,y)\leftarrow p}{\abs{1-\frac{\alpha_y}{q}}\leq1/2}\geq1-2\ve.
\end{eqnarray*}
By the union bound, item 1. follows.

Next we show 2. From item 3. of Lemma \ref{lem:mainlemma},
\begin{eqnarray*}
&&\expec{(x,y)\leftarrow XY}{\relent{M_{xy}}{M_{x}}+\relent{M_{xy}}{M_{y}}}\\
&& = \condmutinf{X}{M}{Y}+\condmutinf{Y}{M}{X} \leq c.
\end{eqnarray*}
Markov's inequality implies $\pr{(x,y)\leftarrow XY}{(x,y)\in G_2}\geq1-\ve.$ Then item 2. follows from the fact that $XY$ and $p$ are $\ve/2$-close.

Finally we show 4. For any $(x,y)\in G_1\cap G_2$,
\begin{align*}
& \relent{M_{xy}}{M_x}  \leq c/\ve \\
& \Rightarrow\pr{m\leftarrow M_{xy}}{\frac{\pr{}{M_{xy}=m}}{\pr{}{M_x=m}}\leq2^{\frac{c/\ve+1}{\ve}}}\geq1-\ve \quad \mbox{(from Fact \ref{fact:markovofrelent})}\\
& \Rightarrow \pr{m\leftarrow M_{xy}}{\frac{u_y(m)\alpha_x}{v_x(m)\alpha_{xy}}\leq2^{\frac{c/\ve+1}{\ve}}}\geq1-\ve \quad \mbox{(from~\eqref{eqn:mxy} and~\eqref{eqn:mx})}\\
& \Rightarrow\pr{m\leftarrow M_{xy}}{\frac{u_y(m)}{v_x(m)}\leq2^{\Delta}}\geq1-\ve . \\
&  \hspace{1in} \mbox{($(x,y)\in G_1$ and the choice of $\Delta$)}
\end{align*}
Similarly, $\pr{m\leftarrow M_{xy}}{\frac{u_x(m)}{v_y(m)}\leq2^{\Delta}}\geq1-\ve.$
By the union bound,
\[\pr{m\leftarrow M_{xy}}{\frac{u_y(m)}{v_x(m)}\leq2^{\Delta}~\text{and}~ \frac{u_x(m)}{v_y(m)}\leq2^{\Delta}}\geq1-2\ve,\]
which implies $(x,y)\in G$. Hence $G_1\cap G_2\subseteq G$.
\end{proof}

Following few claims establish the desired properties of protocol $\Pi'$ (Figure \ref{fig:pro}).

\begin{definition}
Define the following events.
\begin{itemize}

  \item $E$ occurs if the smallest $i\in\A$ satisfies $\mathbf{h}(\mathbf{m_i})=\mathbf{r}$ and $i\in\B$. Note that $E$ implies $\A\neq\emptyset$.

  \item $B_c$ (subevent of $E$) occurs if $E$ occurs and there exist $j\in\B$ such that $\mathbf{h}(\mathbf{m_j})=\mathbf{r}$ and $\mathbf{m_i} \neq \mathbf{m_j}$, where $i$ is the smallest element in $\A$.

  \item $H \defeq E - B_c $.

\end{itemize}
\end{definition}

Below we use conditioning on $(x,y)$ as shorthand for ``Alice's input is $x$ and Bob's input is $y$".
\begin{claim}\label{claim:1}
For any $(x,y)\in G_1\cap G_2$, we have

\begin{enumerate}

\item for all $i\in[T]$,
$$\frac{1}{2}\cdot\frac{q}{|\M|2^{\Delta}}\leq\pr{\mathbf{r}_{\Pi'}}{\text{Alice accepts $\mathbf{m_i}$}|~(x,y)}\leq\frac{3}{2}\cdot\frac{q}{|\M|2^{\Delta}},$$
and
$$\frac{1}{2}\cdot\frac{q}{|\M|2^{\Delta}}\leq\pr{\mathbf{r}_{\Pi'}}{\text{Bob accepts $\mathbf{m_i}$}|~(x,y)}\leq\frac{3}{2}\cdot\frac{q}{|\M|2^{\Delta}},$$ where
$\mathbf{r}_{\Pi'}$ is the internal randomness of protocol $\Pi'$;

\item $\pr{\mathbf{r}_{\Pi'}}{B_c| ~ (x,y), E}\leq\ve;$

\item $\pr{\mathbf{r}_{\Pi'}}{H|~(x,y)}\geq(1-4\ve) \cdot 2^{-k-\Delta-2}.$

\end{enumerate}

\end{claim}

\begin{proof}

1. We do the argument for Alice. Similar argument follows for Bob. Note that $u_x(m), v_x(m)\in[0,1]$. Then for all  $(x,y) \in \X \times \Y$,
\begin{align*}
& \pr{\mathbf{r}_{\Pi'}}{ \text{Alice accepts $\mathbf{m_i}$}|~ (x,y)} =\frac{1}{|\M|}\sum_m\frac{u_x(m)v_x(m)}{2^{\Delta}}=\frac{\alpha_x}{|\M|2^{\Delta}}.
\end{align*}
Item 1 follows by the fact that $(x,y)\in G_1$.

2. Define $E_i$ (subevent of $E$) when $i$ is the smallest element of $\A$.  For all  $(x,y) \in G_1\cap G_2$, we have :
\begin{eqnarray*}
&&\pr{\mathbf{r}_{\Pi'}}{B_c | ~ (x,y), E_i}\\
&&=\pr{\mathbf{r}_{\Pi'}}{\exists j\ : ~j\in\B ~\text{and}~ \mathbf{h}(\mathbf{m_j})=\mathbf{r} ~\text{and}~ \mathbf{m_j} \neq \mathbf{m_i} | ~ (x,y), E_i}\\
&&\leq\sum_{j\in[T], j \neq i }\pr{\mathbf{r}_{\Pi'}}{j\in\B ~\text{and}~ \mathbf{h}(\mathbf{m_j})=\mathbf{r} ~\text{and}~ \mathbf{m_j} \neq \mathbf{m_i} | ~ (x,y), E_i}\\
&& \hspace{4.6cm} \mbox{(from the union bound)}\\
&&\leq\sum_{j\in[T], j \neq i }\pr{\mathbf{r}_{\Pi'}}{j\in\B|~ (x,y) , E_i}\\
&& \cdot \pr{\mathbf{r}_{\Pi'}}{\mathbf{h}(\mathbf{m_j})=\mathbf{r}  | ~ (x,y), E_i, j \in \B, \mathbf{m_j} \neq \mathbf{m_i}} \\
&&\leq T \cdot \frac{3q}{ |\M|2^{\Delta+1}}\cdot \frac{1}{2^{k}} \\
&&\hspace{0.1cm}\mbox{(two-wise independence of $\mathbf{h}$ and item 1. of this Claim)}\\
&&\leq \ve. \quad \mbox{(from choice of parameters)}
\end{eqnarray*}
Since above holds for every $i$, it implies $\pr{\mathbf{r}_{\Pi'}}{B_c | ~ (x,y), E} \leq \ve$.

3. Consider,
\begin{align*}
&\pr{\mathbf{r}_{\Pi'}}{E|~(x,y)}=\pr{\mathbf{r}_{\Pi'}}{\A\neq\varnothing|~(x,y)}\cdot\pr{\mathbf{r}_{\Pi'}}{E|~\A\neq\varnothing,(x,y)}  \\
& \geq \(1 - \(1-\frac{1}{2}\cdot\frac{q}{|\M|2^{\Delta}}\)^T\) \cdot\pr{\mathbf{r}_{\Pi'}}{E|~\A\neq\varnothing,(x,y)}  \\
& \hspace{1.6in}   \mbox{(using item 1. of this Claim)} \\
& \geq (1-\ve)  \cdot\pr{\mathbf{r}_{\Pi'}}{E|~\A\neq\varnothing,(x,y)}  ~\mbox{(from choice of parameters)}\\
&= (1-\ve) \cdot \pr{\mathbf{r}_{\Pi'}}{\mathbf{h}(\mathbf{m_i}) =\mathbf{r}|~\A\neq\varnothing,(x,y)}\cdot\\
&\quad \pr{\mathbf{r}_{\Pi'}}{i\in \B|~i\in\A,~ \mathbf{h}(\mathbf{m_i})=\mathbf{r}, (x,y)} \\
& \mbox{(from here on we condition on $i$ being the first element of $\A$)}\\
&=(1-\ve) \cdot 2^{-k} \cdot \pr{\mathbf{r}_{\Pi'}}{i\in\B|~i\in\A,(x,y)}\\
&=(1-\ve) \cdot 2^{-k} \cdot \frac{\pr{\mathbf{r}_{\Pi'}}{i\in\B ~\text{and}~ i\in\A|~(x,y)}}{\pr{\mathbf{r}_{\Pi'}}{i\in\A|~(x,y)}}\\
&\geq \frac{2}{3q}(1-\ve) \cdot 2^{-k} \cdot |\M|2^{\Delta} \cdot \pr{\mathbf{r}_{\Pi'}}{i\in\B ~\text{and}~ i\in\A|~(x,y)}    \\
& \hspace{1.6in}  \mbox{(using item 1. of this claim)}\\
&=\frac{2}{3q}(1-\ve) \cdot 2^{-k} \cdot |\M|2^{\Delta} \cdot\\
&\sum_{m\in\M}\frac{1}{|\M|2^{2\Delta}}\min\set{u_x(m),2^{\Delta}v_y(m)}\cdot\min\set{u_y(m),2^{\Delta}v_x(m)} \\
& \hspace{1.2in} \mbox{(from construction of protocol $\Pi'$)}\\
&\geq\frac{2}{3q}(1-\ve) \cdot 2^{-k} \cdot |\M|2^{\Delta} \cdot \sum_{m\in G_{xy}}\frac{u_x(m)u_y(m)}{|\M|2^{2\Delta}}\\
& \hspace{0.1in} \mbox{($G_{xy}\defeq\{m:u_x(m)\leq2^{\Delta}v_y(m)~\text{and}~ u_y(m)\leq2^{\Delta}v_x(m)\}$)}\\
&=\frac{2}{3q}(1-\ve) \cdot 2^{-k} \cdot |\M|2^{\Delta} \cdot  \frac{\alpha_{xy}}{|\M|2^{2\Delta}}\sum_{m\in G_{xy}}\frac{u_x(m)u_y(m)}{\alpha_{xy}}\\
&\geq\frac{1}{3}(1-\ve) \cdot 2^{-k-\Delta}\cdot \pr{m\leftarrow M_{xy}}{m\in G_{xy}} \\
& \hspace{1.2in} \mbox{(since $(x,y) \in G_1$ and (\ref{eqn:mxy}))}
\end{align*}
\begin{align*}
&\geq\frac{1}{3}(1-\ve) \cdot 2^{-k-\Delta} \cdot (1-2\ve)  \\
& \hspace{1.2in} \mbox{(since $(x,y) \in G$, using item 4. of Claim \ref{claim:probofg})}\\
&\geq(1-3\ve)\cdot 2^{-k-\Delta-2} .
\end{align*}
Finally, using item 2. of this Claim.
\begin{eqnarray*}
&\pr{\mathbf{r}_{\Pi'}}{H|~(x,y)}=\pr{\mathbf{r}_{\Pi'}}{E|~(x,y)}(1- \pr{\mathbf{r}_{\Pi'}}{B_c |~(x,y),E})\\
&\geq(1-4\ve) \cdot 2^{-k-\Delta-2}.
\end{eqnarray*}
\end{proof}

\begin{claim}\label{claim:probnonabort}
$\pr{p,\mathbf{r}_{\Pi'}}{H}\geq(1-\frac{23}{2}\ve) \cdot 2^{-k-\Delta-2}.$
\end{claim}

\begin{proof}

\begin{eqnarray*}
&&\pr{p,\mathbf{r}_{\Pi'}}{H}\geq\sum_{(x,y)\in G_1\cap G_2}p(x,y)\pr{\mathbf{r}_{\Pi'}}{H|~(x,y)}\\
&&\geq(1-4\ve) \cdot 2^{-k-\Delta-2}\sum_{(x,y)\in G_1\cap G_2}p(x,y)\\
&&\geq(1-\frac{23}{2}\ve) \cdot 2^{-k-\Delta-2}.
\end{eqnarray*}
The second inequality is by Claim \ref{claim:1}, item 3, and the last inequality is by Claim \ref{claim:probofg} item 3.
\end{proof}

The following claim is an important original contribution of this work (not present in the previous works~\cite{Kerenidis2012,Braverman2012,BravermanWeinstein2011}.)
The claim helps us establish a crucial property of $\Pi'$. The property is that the bad inputs $(x,y)$
for which the distribution of $\Pi'$'s sample for $M$, conditioned on non-abort, deviates a lot from the desired,
their probability is nicely reduced in the final distribution of $\Pi'$, conditioned on non-abort.  This helps us
to argue that the joint distribution of inputs and the transcript in $\Pi'$, conditioned on non-abort, is still close in $\ell_1$ distance to $XYM$.

\begin{claim}\label{claim:distclose}
Let $AB$ and $A'B'$ be random variables over $\A_1\times\B_1$ and $h:\A_1\rightarrow [0,+\infty)$ be a function. Suppose
for any $a\in\A_1$, there exist functions $f_a,g_a:\B_1\rightarrow[0,+\infty)$, such that
\begin{enumerate}
  \item $\sum_{a,b} h(a)f_a(b)=1$, and $\pr{}{AB=ab}=h(a)f_a(b)$;

  \item $f_a(b)\geq g_a(b),$ for all $(a,b)\in\A_1\times\B_1$;

  \item $\pr{}{A'B'=ab}=h(a)g_a(b)/C,$ where $C=\sum_{a,b}h(a)g_a(b)$;

  \item $\pr{a\leftarrow A}{\pr{b\leftarrow B_a}{f_a(b)=g_a(b)}\geq1-\delta_1}\geq1-\delta_2$, for $\delta_1\in[0,1), \delta_2\in[0,1)$.

\end{enumerate}

Then $\onenorm{AB-A'B'}\leq\delta_1+\delta_2$.

\end{claim}


\begin{proof} Set $G\defeq\set{(a,b):f_a(b)=g_a(b)}$. By condition 4, $\pr{(a,b)\leftarrow AB}{(a,b)\in G}\geq1-\delta_1-\delta_2.$ Then
\begin{align}
C&=\sum_{a,b}h(a)g_a(b)\geq\sum_{a,b:(a,b)\in G}h(a)f_a(b)\nonumber\\
&=\pr{(a,b)\leftarrow AB}{(a,b)\in G}\geq1-\delta_1-\delta_2.\label{eq:Clarge}
\end{align}
We have
\begin{eqnarray*}
&&\onenorm{AB-A'B'}=\frac{1}{2}\sum_{a,b}|h(a)f_a(b)-\frac{1}{C}h(a)g_a(b)|\\
&\leq&\frac{1}{2}\sum_{a,b}\br{|h(a)f_a(b)-h(a)g_a(b)|+|h(a)g_a(b)-\frac{1}{C}h(a)g_a(b)|}\\
&\leq&\frac{1}{2}\br{\sum_{a,b}\(h(a)f_a(b)-h(a)g_a(b)\)+\frac{1-C}{C}\sum_{a,b}h(a)g_a(b)} \\
&&\hspace{1.5in} \mbox{(using item 2. of this claim)}\\
&\leq&\frac{1}{2}\br{\sum_{a,b:(a,b)\notin G}h(a)f_a(b)+1-C}\\
&=&\frac{1}{2}\br{\pr{(a,b)\leftarrow AB}{(a,b)\notin G}+1-C}\leq\delta_1+\delta_2 ~~ \mbox{(from~\eqref{eq:Clarge})}
\end{eqnarray*}
\end{proof}

%

\begin{claim}\label{claim:singlemessagecloseness}
Let the input of protocol $\Pi'$ be drawn according to $p$. Let $X^1Y^1M^1$  represent
the input and the transcript (the part of the public coins drawn from $\M$) conditioned on $H$. Then we have
$\onenorm{XYM-X^1Y^1M^1}\le10\ve$. Note that this implies that $\onenorm{X^1Y^1A^1B^1-XY\tilde{M}\tilde{M}}\leq 10\ve$,
where $\tilde{M}$ represents the last $\lceil\log|\Z|\rceil$ bits of $M$ and $A^1, B^1$ represent outputs of
Alice and Bob respectively, conditioned on $H$.
\end{claim}

\begin{proof}

For any $(x,y)$, define
$$w_{xy}(m)\defeq\min\set{u_x(m),2^{\Delta}v_y(m)}\cdot\min
\set{u_y(m),2^{\Delta}v_x(m)}.$$
From step 2 (a),(b),(c), of protocol $\Pi'$,
$\pr{}{M^1X^1Y^1=mxy}=\frac{1}{C}p(x,y)w_{xy}(m)$, where $C=\sum_{xym}p(x,y)w_{xy}(m)$.
Now,
\begin{eqnarray*}
&\pr{(x,y)\leftarrow XY}{\pr{m\leftarrow M_{xy}}{w_{xy}(m)=u_x(m)u_y(m)}\geq1-2\ve}\\
&=\pr{(x,y)\leftarrow XY}{(x,y)\in G}\geq1-8\ve.
\end{eqnarray*}
The last inequality above follows using items 3. and 4. of  Claim \ref{claim:probofg} and the fact that $XY$ and $p$ are $\ve/2$-close.

Finally using Claim \ref{claim:distclose} (by substituting $\delta_1 \leftarrow 2 \ve, \delta_2 \leftarrow 8 \ve, A \leftarrow XY, B \leftarrow M, A' \leftarrow X^1Y^1, B' \leftarrow M^1, h \leftarrow \frac{p}{q}, f_{(x,y)}(m) \leftarrow u_x(m) u_y(m)$
and $g_{(x,y)}(m) \leftarrow w_{xy}(m) $), we get that \newline $\onenorm{X^1Y^1M^1-XYM}\leq10\ve$.
\end{proof}

 We are now ready to finish the proof of Lemma~\ref{lem:mainlemma}. Consider the protocol $\Pi'$. We claim that it satisfies
Lemma \ref{lem:zeroprotocolimpliesrec} by taking the correspondence between quantities in Lemma \ref{lem:zeroprotocolimpliesrec}
and Lemma~\ref{lem:mainlemma} as follows : $c \leftarrow (c/\ve^2+3/\ve), \ve \leftarrow 11\ve, \beta \leftarrow \beta, \delta \leftarrow \delta, z \leftarrow z, X'Y' \leftarrow p$.

Item 1. of   Lemma \ref{lem:zeroprotocolimpliesrec} is implied by Claim \ref{claim:probnonabort}
since  $(1-\frac{23}{2}\ve) \cdot 2^{-k-\Delta-2}\geq2^{-(c/\ve^2+3/\ve)}$, from choice of parameters.

Item 2. of   Lemma \ref{lem:zeroprotocolimpliesrec} is implied since $\onenorm{X^1Y^1-p}\leq\onenorm{X^1Y^1-XY}+\onenorm{XY-p}\leq \frac{21}{2}\ve$,
using item 2. of Lemma \ref{lem:mainlemma}, Fact \ref{fact:one norm and rel ent}  and
Claim \ref{claim:singlemessagecloseness}.

Item 3. of   Lemma \ref{lem:zeroprotocolimpliesrec} is implied since
$\err{f}{X^1Y^1M^1}\leq\err{f}{XYM}+\onenorm{X^1Y^1M^1-XYM}\leq 11\ve$, using  item 4. in Lemma \ref{lem:mainlemma} and Claim \ref{claim:singlemessagecloseness}.

This implies
$$\lsrec{z,p}{(1+\ve')\delta/\beta,\delta}{f}<\frac{c/\ve^2 + 3/\ve}{11\ve} \leq \frac{2c}{11	\ve^3}. \qed$$
\end{proof}

We can now prove our main result. 
\begin{theorem}\label{thm:sdptsrec1}
Let $\X,\Y,\Z$ be finite sets, $f\subseteq\X\times\Y\times\Z$ be a relation, and $t>1$ be an integer. Let $\mu$ be a distribution on $\X\times\Y$. Let $z \in \Z$ and $\beta\defeq\pr{(x,y)\leftarrow \mu}{f(x,y)=\set{z}}$.  Let $0<\ve<1/3$ and $\ve',\delta>0$ be such that $\frac{\delta+22\ve}{\beta-33\ve} < (1+\ve')\frac{\delta}{\beta}$. It holds that,
\[\mathrm{R}^{\mathrm{pub}}_{1-(1-\ve)^{\lfloor\ve^2 t/32\rfloor}}(f^t)\geq\frac{\ve^2}{32} \cdot t \cdot \br{ 11 \ve \cdot \lsrec{z,\mu}{(1+\ve')\delta/\beta,\delta}{f}  - 2}.\]
\end{theorem}

\begin{proof}

Set $\delta_1\defeq\ve^2/32$. 
define 
$$c\defeq 11 \ve \cdot  \lsrec{z,\mu}{(1+\ve')\delta/\beta,\delta}{f}  - 2$$
and $XY\sim\mu^k .$
By Fact~\ref{fact:yaos principle}, it suffices to show 
$$\mathrm{D}^{\mu^t}_{1-(1-\ve)^{\lfloor\ve^2 t/32\rfloor}}(f^t)\geq\delta_1tc.$$  Let $\Pi$ be a
deterministic two-way communication protocol, that computes $f^t$, with total communication $\delta_1 ct$ bits.
The following claim implies that the success of $\Pi$ is at most $(1-\ve)^{\lfloor\delta_1 t\rfloor}$, and this shows the desired.
\end{proof}

\begin{claim}
For each $i\in[t]$, define a binary random variable $T_i\in\set{0,1}$, which represents the success of $\Pi$ on the
$i$-th instance. That is, $T_i=1$ if the protocol computes the $i$-th instance of $f$ correctly, and $T_i=0$ otherwise.
Let $t'\defeq\lfloor\delta_1 t\rfloor$. There exists $t'$ coordinates $\set{i_1,\cdots,i_{t'}}$ such that for each $1\leq r\leq t'-1$,
\begin{enumerate}
\item either $\pr{}{T^{(r)}=1}\leq(1-\ve)^{t'}$  or
\item $\pr{}{T_{i_{r+1}}=1|~T^{(r)}=1}\leq1-\ve$, where $T^{(r)}\defeq\prod_{j=1}^rT_{i_j}$.
\end{enumerate}
\end{claim}

\begin{proof}

Suppose we have already identified $r$ coordinates, $i_1,\cdots,i_r$ satisfying that $\pr{}{T_{i_1}}\leq1-\ve$ and $\pr{}{T_{i_{j+1}}=1|~T^{(j)}=1}\leq1-\ve$
for $1\leq j\leq r-1.$ If $\pr{}{T^{(r)}=1}\leq(1-\ve)^{t'}$, then we are done. So from now on we assume $\pr{}{T^{(r)}=1}>(1-\ve)^{t'}\geq2^{-\delta_1 t}$. Here we assume  $r\geq1$.
Similar arguments also work when $r=0$, that is  for identifying the first coordinate, which we skip for the sake of avoiding
repetition.

Let $D$ be a random variable uniformly distributed in $\set{0,1}^t$ and independent of $XY$. Let $U_i=X_i$ if $D_i=0$, and
$U_i=Y_i$ if $D_i=1$. For any random variable $L$, define $L^1\defeq(L|T^{(r)}=1).$ If $L=L_1\cdots L_t$, define $L_{-i}\defeq L_1\cdots L_{i-1}L_{i+1}\cdots L_t$.
Let $\C\defeq\set{i_1,\cdots, i_r}.$ Define $R_i\defeq D_{-i}U_{-i}X_{\C\cup[i-1]}Y_{\C\cup[i-1]}$.

Now let us  apply Lemma \ref{lem:mainlemma} by substituting $XY \leftarrow X^1_jY^1_j, M \leftarrow R_j^1M^1, p\leftarrow X_jY_j, z\leftarrow z, \ve \leftarrow \ve, \delta \leftarrow \delta, \beta \leftarrow \beta, \ve' \leftarrow \ve' $ and $c \leftarrow 16 \delta_1 (c+1)$. Condition 1. in Lemma \ref{lem:mainlemma} is  implied by Claim \ref{claim:distribution}. Conditions 2. and 3. are implied by Claim \ref{claim:goodcoordinate}. Also we have $\lsrec{z,\mu}{(1+\ve')\delta/\beta,\delta}{f}>\frac{32\delta_1 (c+1)}{11\ve^{3}}$, by our choice of $c$. Hence condition 4. must be false and hence $\err{f}{X^1_jY^1_jM^1}= \err{f}{X^1_jY^1_jR_j^1M^1}>\ve$. This shows condition 2. of this Claim.
\end{proof}

\begin{claim}\label{claim:distribution}
Let $\R$ denote the space of $R_j$. There exist functions $u_{x_j},u_{y_j}:\R\times\M\rightarrow[0,1]$ for all $(x_j,y_j)\in\X\times\Y$ and a real number $q>0$ such that
\[\pr{}{X^1_jY^1_jR^1_jM^1=x_jy_jr_jm}=\frac{1}{q}\mu(x_j,y_j)u_{x_j}(r_j,m)u_{y_j}(r_j,m).\]
\end{claim}

\begin{proof}
Note that $X_jY_j$ is independent of $R_j$.  Now consider a private-coin two-way protocol $\Pi_1$ with input $X_jY_j$ as follows. Let Alice generate $R_j$ and  send to Bob. Alice and Bob then generate $(X_{-j})_{x_jr_j}$ and $(Y_{-j})_{y_jr_j}$, respectively. Then they run the protocol $\Pi$. Thus, from Lemma~\ref{lem:privateccdistribution},
\[\pr{}{X_jY_jR_jM=xy_jrm}=\mu(x_j,y_j) \cdot v_{x_j}(r_j, m) \cdot v_{y_j}(r_j, m),\]
where $v_{x_j}, v_{y_j}:\R \times \M\rightarrow[0,1]$, for all $(x_j,y_j)\in\X\times\Y$.

Note that conditioning on $T^{(r)}=1$ corresponds to choosing a subset, say S, of $\R \times \M$. Let
$$q \defeq \sum_{x_jy_jr_jm: (r_j,m) \in S}\mu(x_j,y_j)v_{x_j}(r_j,m)v_{y_j}(r_j,m) \enspace .$$
Then
$$\pr{}{X_j^1Y_j^1R_j^1M^1=x_jy_jr_jm} = \frac{1}{q}\mu(x_j,y_j)v_{x_j}(r_j,m)v_{y_j}(r_j,m),$$
for $(r_j,m) \in S$ and $\pr{}{X_j^1Y_j^1R_j^1M^1=x_jy_jr_jm} = 0$ otherwise.

Now define
$$u_{x_j}(r_j,m)\defeq v_{x_j}(r_j,m),~\text{and}~ u_{y_j}(r_j,m)\defeq v_{y_j}(r_j,m) ,$$
for $(r_j,m) \in S$ and define them to be $0$ otherwise. The claim follows.
\end{proof}

\begin{claim}\label{claim:goodcoordinate}

If $\pr{}{T^{(r)}=1}>2^{-\delta_1 t}$, then there exists a coordinate $j\notin \C$ such that
\begin{eqnarray}
&\relent{X^1_jY^1_j}{X_jY_j}\leq8\delta_1 = \frac{\ve^2}{4}, \label{eqn:distclose} \\
~\text{and}\nonumber\\
&\condmutinf{X^1_j}{M^1R^1_j}{Y^1_j}+\condmutinf{Y^1_j}{M^1R^1_j}{X^1_j}\leq16\delta_1(c+1).\nonumber\\
~\hspace{2mm}\label{eqn:condmutinf}
\end{eqnarray}
\end{claim}

\begin{proof}
This follows using Claim III.6 in \cite{Jain2012}. We include a proof in Appendix for completeness.
\end{proof}

\subsection*{Conclusion and open problems}\label{sec:conclusion}
We provide a strong direct product result for the two-way public-coin
communication complexity in terms of an important and widely used lower bound method, the smooth rectangle bound. Some natural questions that arise are:
\begin{enumerate}
\item Is the  smooth rectangle bound a tight lower bound for  the  two-way public-coin
communication complexity for all relations? If yes, this would imply a strong direct product
result for the  two-way public-coin communication complexity for all relations, settling a major open question in this area.
To start with we can ask: Is the smooth rectangle bound a polynomially tight lower bound for the two-way public-coin communication complexity for all relations?
\item Or on the other hand, can we exhibit a relation for which the  smooth rectangle bound is (asymptotically) strictly smaller than its  two-way public-coin communication complexity?
\item Can we show similar direct product results in terms of possibly stronger lower bound methods like the partition bound and the information complexity?
\item It will be interesting to obtain new optimal lower bounds for interesting functions and relations using  the smooth rectangle bound, implying strong direct product results for them.
\end{enumerate}


\noindent {\it Acknowledgement}. We thank Prahladh Harsha for helpful discussions.

%
\bibliographystyle{abbrv}

%
%
\appendix

\label{sec:deferredproofs}

\begin{proofof}{Lemma~\ref{lem:ratiovs1}}
Let $G=\set{a:\abs{1-\frac{\pr{}{A'=a}}{\pr{}{A=a}}}\leq\frac{\ve}{r}}$, then
\begin{eqnarray*}
&&2\ve\geq\sum_a\abs{\pr{}{A=a}-\pr{}{A'=a}}\\
&&\geq\sum_{a\not\in G}\abs{\pr{}{A=a}-\pr{}{A'=a}}\\
&&=\sum_{a\not\in G}\pr{}{A=a}\abs{1-\frac{\pr{}{A'=a}}{\pr{}{A=a}}}\geq\pr{a\leftarrow A}{a\not\in G}\cdot\frac{\ve}{r}.
\end{eqnarray*}
Thus $\pr{a\leftarrow A}{a\in G}\geq 1-2r$. The second inequality follows immediately.
\end{proofof}

\begin{proofof}{Lemma~\ref{lem:Dgeqsrec}} Let $c \defeq \lsrec{z,\lambda}{(1+\ve') \frac{\delta}{\beta},\delta} {f}$. Let $g$ be such that $\lrec{z,\lambda}{(1+\ve') \frac{\delta}{\beta}}{g} = c$ and  
$$\pr{(x,y)\leftarrow\lambda}{f(x,y)\neq g(x,y)}\leq\delta.$$
 If $\mathrm{D}^{\lambda}_{\ve}(f) \geq c - \log(4/\ve)$ then we are done using Fact~\ref{fact:yaos principle}.

So lets assume for contradiction that $\mathrm{D}^{\lambda}_{\ve}(f) < c - \log(4/\ve)$. This implies that there exists a  deterministic
protocol $\Pi$ for $f$ with communication $c-\log(4/\ve)$ and distributional error under $\lambda$ bounded by $\ve$.
Since $$\pr{(x,y)\leftarrow\lambda}{f(x,y)\neq g(x,y)}\leq\delta,$$
the protocol $\Pi$ will have distributional error at most $\ve+\delta$ for $g$. Let $M$
represent the message transcript of $\Pi$ and let $O$ represent protocol's output. We assume that the last $\lceil \log |Z| \rceil$ bits of $M$ contain $O$. We have,
\begin{enumerate}
\item $\pr{m \leftarrow M}{\pr{}{M=m}\leq2^{-c}} \leq \ve/4 ,$ since the total number of message transcripts in $\Pi$ is at most $2^{c - \log(4/\ve)}$.
\item $\pr{m \leftarrow M}{O=z|~M=m} > \beta - \ve ,$ \newline since   $ \pr{(x,y)\leftarrow\lambda}{f(x,y) = \set{z}}=\beta$ and distributional error of $\Pi$  under $\lambda$ is bounded by $\ve$ for $f$.
\item  $\pr{m \leftarrow M}{\pr{(x,y) \leftarrow (XY)_m}{(x,y,O) \notin g|~M=m} \geq  \frac{\ve+\delta}{\beta-2\ve}} \leq \beta - 2\ve$, since distributional error of $\Pi$  under $\lambda$ is bounded by $\ve+\delta$ for $g$.
\end{enumerate}
Using all of above we obtain a message transcript $m$ such that $\pr{}{M=m} > 2^{-c}$ and  $(O=z|~M=m)$ and
\begin{align*}
\pr{(x,y) \leftarrow (XY|M=m)}{(x,y,O) \notin g|~M=m} &\leq  \frac{\ve+\delta}{\beta-2\ve} \\
& < (1+\ve') \frac{\delta}{\beta}.
\end{align*}
This and the fact that  the support of $(XY|~M=m)$ is a rectangle, implies that $\lrec{z,\lambda}{(1+\ve') \frac{\delta}{\beta}}{g} < c$, contradicting the definition of $c$. Hence it must be that  $\mathrm{D}^{\lambda}_{\ve}(f) \geq c - \log(4/\ve)$, which using Fact~\ref{fact:yaos principle} shows the desired.
\end{proofof}

\begin{proofof}{Claim~\ref{claim:goodcoordinate}} The following calculations are helpful for achieving (\ref{eqn:distclose}).
\begin{align}   \label{eqn:one coordinate close}
    \delta_1 k & > \rminent{X^1Y^1}{XY}        \geq
    \relent{X^1Y^1}{XY}\nonumber
    \\ &\geq  \sum_{i\notin C} \relent{X^1_iY^1_i}{X_i Y_i} ,
\end{align}
where the first inequality follows from the assumption that
$\prob{T^{(r)}=1} > 2^{-\delta k}$, and the last inequality follows from Fact~\ref{fact:relative entropy splitting}.
The following calculations are helpful for (\ref{eqn:condmutinf}).
\begin{align}
    &\delta_1 k > \rminent{X^1Y^1D^1U^1}{XYDU} \nonumber\\
    &\geq \relent{X^1Y^1D^1U^1}{XYDU} \nonumber \\
    &\geq \expec{(d,u,x_C,y_C)\\\leftarrow D^1,U^1,X^1_C,Y^1_C}
        {\relent{\br{X^1Y^1}_{d,u,x_C,y_C}}{\br{XY}_{d,u,x_C,y_C}}}
        \label{eqn:eq2} \\
    &= \sum_{i \notin C} \; \underset{\substack{(d,u,x_{C\cup[i-1]},y_{C\cup[i-1]})\\\leftarrow
        D^1,U^1,X_{C\cup[i-1]}^1,Y_{C\cup[i-1]}^1}}{\mathbb{E}} \nonumber \\
    &\ \Big[\relent{\br{X_i^1Y_i^1}_{\substack{d,u,x_{C\cup[i-1]}, \\ y_{C\cup[i-1]}}}}
        {\br{X_iY_i}_{\substack{d,u,x_{C\cup[i-1]}, \\ y_{C\cup[i-1]}}}} \Big]
        \label{eqn:eq3}\\
    &=\sum_{i\notin C}
        \expec{(d_i,u_i,r_i) \\\leftarrow D^1_i,U^1_i,R^1_i}
        { \relent{(X^1_iY^1_i)_{d_i,u_i,r_i}}{(X_iY_i)_{d_i,u_i,r_i}}}
        \label{eqn:eq4}\\
    &= \frac{1}{2} \sum_{i \notin C} \;
        \expec{(r_i,x_i)\leftarrow R^1_i,X^1_i}
        {\relent{\br{Y_i^1}_{r_i, x_i}}
        {\br{Y_i}_{x_i}}} +\nonumber\\
        &\quad  \frac{1}{2} \sum_{i \notin C} \;
    \expec{(r_i,y_i)\leftarrow R^1_i,Y^1_i}
    {\relent{\br{X_i^1}_{r_i, y_i}}
    {\br{X_i}_{y_i}}} .
        \label{eqn:junk is independent for Y}
\end{align}
Above, Eq.~\eqref{eqn:eq2} and Eq.~\eqref{eqn:eq3} follow from
Fact~\ref{fact:relative entropy splitting}; Eq.~\eqref{eqn:eq4} is
from the definition of $R_i$. Eq.~\eqref{eqn:junk is independent for
Y} follows since $D^1_i$ is independent of $R^1_i$ and with
probability half $ D^1_i$ is $0$, in which case $U^1_i = X^1_i$  and
with probability half $ D^1_i$ is $1$ in which case $U_i^1 = Y_i^1$.
By Fact \ref{fact:12},
\begin{equation}\label{eqn:correlatedsample}
\sum_{i\notin C}\br{\condmutinf{X^1_i}{R^1_i}{Y^1_i}+\condmutinf{Y^1_i}{R^1_i}{X^1_i}}\leq2\delta_1 k.
\end{equation}

We also need the following calculations, which exhibits that the
information carried by messages about sender's input is small.

\begin{align}
  &   \delta_1 c k \geq \abs{M^1}\geq \condmutinf{X^1Y^1}{M^1}{D^1U^1X^1_CY^1_C}\nonumber\\
  & =\sum_{i\notin C} \condmutinf{X^1_iY^1_i}{M^1}
    {D^1U^1X^1_{C\cup[i-1]}Y^1_{C\cup[i-1]}}
    \nonumber \\
    & = \sum_{i\notin
    C}\condmutinf{X^1_iY^1_i}{M^1}{D^1_iU^1_iR^1_i}\nonumber\\
    & = \frac{1}{2}\sum_{i\notin C}\br{\condmutinf{X^1_i}{M^1}{R^1_iY^1_i}+\condmutinf{Y^1_i}{M^1}{R^1_iX^1_i}}.
    \label{eqn:mutual inf XY and M}
\end{align}

Above first equality follows from chain rule for mutual information, second equality follows from definition of $R_i^1$ and the third equality follows since with
probability half $ D^1_i$ is $0$, in which case $U^1_i = X^1_i$  and
with probability half $ D^1_i$ is $1$ in which case $U_i^1 = Y_i^1$.

Combining Eqs.~\eqref{eqn:one coordinate close}, \eqref{eqn:junk
is independent for Y} and \eqref{eqn:mutual inf XY and
M}, and making standard use of Markov's
inequality, we can get a coordinate $j\notin C$ such that
\begin{eqnarray*}
&&\relent{X^1_jY^1_j}{X_jY_j}\leq8\delta_1,\\
&&\condmutinf{X^1_j}{R^1_j}{Y^1_j}+\condmutinf{Y^1_j}{R^1_j}{X^1_j}\leq16\delta_1,\\
&&\condmutinf{X^1_j}{M^1}{R^1_jY^1_j}+\condmutinf{Y^1_j}{M^1}{R^1_jX^1_j}\leq16\delta_1c.
\end{eqnarray*}
The first inequality is exactly the same as Eq.~\eqref{eqn:distclose}. Eq.~\eqref{eqn:condmutinf} follows by adding the last two inequalities.
\end{proofof}

\subsection*{Alternate definition of smooth rectangle bound}
An alternate definition of the smooth rectangle bound was introduced by Jain and
Klauck~\cite{JainKlauck2010}, using  the following linear program.

\begin{definition} \label{def:srec}
For total function $f:\X \times \Y \rightarrow \Z$, the $\epsilon$- smooth rectangle bound
of $f$ denoted  $\srec{}{\epsilon}{f}$ is defined to be 
$\max\{\srec{z}{\epsilon}{f}: z\in \Z\},$ where $\srec{z}{\epsilon}{f}$ is given by the optimal value of the following
linear program. 

{\footnotesize
\vspace{0.1in}
    \centerline{\underline{Primal}}
    \begin{align*}
      & \text{min:}\quad  \sum_{W \in \W}  v_{W} \\
       \quad &  \forall (x,y) \in  f^{-1}(z): \sum_{W: (x,y) \in W} v_{W} \geq 1 - \epsilon,\\
      & \forall (x,y) \in  f^{-1}(z): \sum_{W: (x,y) \in W} v_{W} \leq 1,\\
      &  \forall (x,y) \notin   f^{-1}(z): \sum_{W: (x,y) \in W} v_{W} \leq \epsilon,\\
      & \forall W : v_{W} \geq 0 \enspace .
    \end{align*}
}

{\footnotesize
    \centerline{\underline{Dual}}
    \begin{align*}
      & \text{max:}\quad   \sum_{(x,y)\in f^{-1}(z)} \left((1-\epsilon) \lambda_{x,y} - \phi_{x,y} \right)- \sum_{(x,y)\notin  f^{-1}(z)} \epsilon \cdot \lambda_{x,y}\\
       \quad &   \forall W : \sum_{(x,y)\in f^{-1}(z)\cap W} (\lambda_{x,y} - \phi_{x,y}) - \sum_{(x,y)\in (W  - f^{-1}(z))} \lambda_{x,y} \leq 1,\\
      & \forall (x,y) : \lambda_{x,y} \geq 0 ; \phi_{x,y} \geq 0 \enspace .
    \end{align*}
}
\end{definition}

\begin{lemma}\label{lem:eqv}
Let $f:\X \times \Y \rightarrow \Z$ be a total function. Let $z \in Z$ and $\ve >0$. There exists a distribution $\mu \in \X \times \Y$ and $\delta, \beta >0$ such that  
$$\lsrec{z,\mu}{(1+\ve^2)\frac{\delta}{\beta},\delta}{f}\geq \log(\srec{z}{\ve}{f}) + 3\log \ve.$$
\end{lemma}
\begin{proof}
Let $(\lambda'_{x,y}, \phi'_{x,y})$ be an optimal solution to the Dual.  For $(x,y) \in f^{-1}(z)$, if $\lambda'_{x,y} > \phi'_{x,y}$ define $\lambda= \lambda'_{x,y} - \phi'_{x,y}$ and $\phi_{x,y} = 0$. Otherwise define $\lambda= 0$ and $\phi_{x,y} = \phi'_{x,y} -  \lambda'_{x,y} $. For  $(x,y) \notin f^{-1}(z)$ define $\phi_{x,y}=0$. We note that $(\lambda_{x,y}, \phi_{x,y})$ is an optimal solution to the Dual with potentially higher objective value. Hence $(\lambda_{x,y}, \phi_{x,y})$ is also an optimal solution to the Dual.

Let us define three sets 
$$U_1 \defeq \{(x,y) |~ f(x,y) = z, \lambda_{x,y} > 0 \},$$
$$U_2 \defeq \{(x,y) |~ f(x,y) = z, \phi_{x,y} > 0 \},$$
$$U_0 \defeq \{(x,y) |~ f(x,y) \neq z, \lambda_{x,y} > 0 \}.$$
Define, 
$$\forall (x,y) \in U_1 : \mu'(x,y) \defeq \lambda_{x,y},$$
$$\forall (x,y) \in U_2 : \mu'(x,y) \defeq \ve \phi_{x,y},$$
$$\forall (x,y) \in U_0 : \mu'(x,y) \defeq \ve\lambda_{x,y}.$$
Define $r \defeq \sum_{x,y} \mu'(x,y)$ and define probability distribution $\mu \defeq \mu'/r$. Let $\srec{z}{\epsilon}{f}= 2^c$.  Define function $g$ such that $g(x,y)= z$ for  $(x,y) \in U_1$; $g(x,y) = f(x,y)$ for $(x,y) \in U_0$ and $g(x,y) = z'$ (for some $z' \neq z$) for $(x,y) \in U_2$. Then,
\begin{align*}
2^c & = \sum_{(x,y)\in f^{-1}(z)} \left((1-\epsilon) \lambda_{x,y} - \phi_{x,y} \right)- \sum_{(x,y)\notin  f^{-1}(z)} \epsilon \cdot \lambda_{x,y}\\
& = (1-\epsilon) \mu'(U_1) - \frac{1}{\ve} \mu'(U_2) -  \mu'(U_0)
\end{align*}
This implies $r \geq \mu'(U_1) \geq 2^c$.
Consider rectangle $W$.
\begin{align*}
& \sum_{(x,y)\in f^{-1}(z)\cap W} (\lambda_{x,y} - \phi_{x,y}) - \sum_{(x,y)\in (W  - f^{-1}(z))} \lambda_{x,y} \leq 1\\
& \Rightarrow \sum_{(x,y)\in U_1\cap W} \mu_{x,y} - \frac{1}{\ve}\sum_{(x,y)\in U_2\cap W}\mu_{x,y} -  \sum_{(x,y)\in U_0\cap W} \frac{1}{\ve} \mu_{x,y} \leq \frac{1}{r}\\
& \Rightarrow \ve\left(\sum_{(x,y)\in U_1\cap W} \mu_{x,y} - \frac{1}{r}\right) \leq \sum_{(x,y)\in U_2\cap W}\mu_{x,y} +  \sum_{(x,y)\in U_0\cap W}  \mu_{x,y} \\
& \Rightarrow \ve\left(\sum_{(x,y)\in g^{-1}(z)\cap W} \mu_{x,y} - \frac{1}{r}\right) \leq \sum_{(x,y)\in W - g^{-1}(z)} \mu_{x,y} \\
& \Rightarrow \ve\left(\sum_{(x,y)\in W} \mu_{x,y} - \frac{1}{r}\right) \leq (1+\ve) \cdot \sum_{(x,y)\in W - g^{-1}(z)} \mu_{x,y} \\
& \Rightarrow \ve\left(\sum_{(x,y)\in W} \mu_{x,y} - 2^{-c} \right) \leq (1+\ve) \cdot \sum_{(x,y)\in W - g^{-1}(z)} \mu_{x,y} .
\end{align*}
Now consider a $W$ with $\mu(W)\geq 2^{-c}/\ve^3$.  We have $\mu(W - g^{-1}(z)) \geq \frac{(1-\ve^3)\ve}{1+\ve} \mu(W) $. 
Define $\beta \defeq \mu(U_1 \cup U_2)$, $\delta \defeq \mu(U_2)$. Now,
\begin{align*}
(1-\ve) r \beta \geq  (1-\ve) \mu'(U_1) \geq \frac{1}{\ve} \mu'(U_2) = \frac{1}{\ve} r \delta .
\end{align*}
Hence we have 
$$\mu(W - g^{-1}(z)) \geq \frac{(1-\ve^3)\delta}{(1 -\ve^2)\beta} \mu(W) \geq  (1+\ve^2)\frac{\delta}{\beta} \mu(W).$$
This implies 
$\lrec{z,\mu}{(1+\ve^2)\delta/\beta}{g} \geq c + 3 \log \ve. $ This implies that 
$$\lsrec{z,\mu}{(1+\ve^2)\frac{\delta}{\beta},\delta}{f}\geq c + 3\log \ve = \log (\srec{z}{\ve}{f})+ 3\log \ve .$$
\end{proof}

\end{document}